\documentclass[10pt, letterpaper]{IEEEtran}

\usepackage{amsmath,amssymb,amsthm}
\usepackage{subfigure}
\usepackage{tikz}
\usepackage[margin=1.4cm]{geometry}
\usepackage{algorithm}
\usetikzlibrary{trees}
\usepackage{epsfig}
\usepackage[noend]{algorithmic}
\usepackage{algorithmic}

\usepackage{url}

\newtheorem{theorem}{Theorem}

\newtheorem{proposition}{Proposition}
\newtheorem{definition}{Definition}
\newtheorem{lemma}{Lemma}

\newtheorem{remark}{Remark}

\begin{document}
\title{\huge Dynamic Monopolies in Colored Tori}
\author{\IEEEauthorblockN{Sara Brunetti\IEEEauthorrefmark{1}, Elena Lodi\IEEEauthorrefmark{1}, Walter Quattrociocchi\IEEEauthorrefmark{1}}\bigskip

\IEEEauthorblockA{\IEEEauthorrefmark{1}
  University of Siena, Italy,\\
  {\tt \small [lodi, sara.brunetti,walter.quattrociocchi]@unisi.it}}\\
}
\date{}
\maketitle
\begin{abstract}

The {\em information diffusion} has been modeled as the spread of an information within a group through a process of social influence, where the diffusion is driven by the so called {\em influential network}. Such a process, which has been intensively studied under the name of {\em viral marketing}, has the goal to select an initial good set of individuals that will promote a new idea (or message) by spreading the "rumor" within the entire social network through the word-of-mouth. 
Several studies used the {\em linear threshold model} where the group is represented by a graph, nodes have two possible states (active, non-active), and the threshold triggering the adoption (activation) of a new idea to a node is given by the number of the active neighbors. 

The problem of detecting in a graph the presence of the minimal number of nodes that will be able to activate the entire network is called {\em target set selection} (TSS). In this paper we extend TSS by allowing nodes to have more than two colors. 
The multicolored version of the TSS can be described as follows: let $G$ be a torus where every node is assigned a color from a finite set of colors. At each local time step, each node can recolor itself, depending on the local configurations, with the color held by the majority of its neighbors. 

We study the initial distributions of colors leading the system to a monochromatic configuration of color $k$, focusing on the minimum number of initial $k$-colored nodes. We conclude the paper by providing the time complexity to achieve the monochromatic configuration.

\end{abstract}


\section{Introduction}
The emerging global effects of protocols based on local interactions are studied in the distributed algorithms field. For instance, the interplay between the election/consensus problems and the protocols based on majority rules has been extensively studied in several works such as \cite{Kossmann} \cite{Peleg96b}, \cite{Thomas79}, \cite{Burman09}, \cite{Shao09} and \cite{Garcia85}. In fact, social behaviors often provide useful insights in determining the way problems are solved in computer science. For instance, matters of security in decision making or in collaborative content filtering/production are enforced by implementing reputation and trust strategies \cite{JosangQ09,schillo99,Sabater2007,Konig2009,quattrociocchi2008b,QuattrociocchiPC09}.
The {\em information diffusion} has been modeled as the spread of an information within a group through a process of social influence, where the diffusion is driven by the so called {\em influential network} \cite{Kats2005}. Such a process, which has been intensively studied under the name of {\em viral marketing} (see for instance \cite{Domingos2001}), has the goal to select an initial good set of individuals that will promote a new idea (or message) by spreading the "rumor" within the entire social network through the word-of-mouth. Such a problem has been extended also to opinion dynamics in social network \cite{amblard01,quattrociocchi2010d}. The first computational study about information diffusion \cite{Granovetter85} used the {\em linear threshold model} where the group is represented by a graph and the threshold triggering the adoption (activation) of a new idea to a node is given by the number of the active neighbors. 
The impossibility of a node to return (or not) in its initial state determines the monotone (or non-monotone) behavior of the activation process. In a graph detecting the presence of the minimal number of nodes that will activate the whole network, namely the target set selection process (TSS), has been proved to be NP-hard through a reduction to the vertex cover \cite{Kempe03}.
In \cite{Chang09a} the maximum size of a {\em minimum perfect target set} under simple and strong majority thresholds - e.g. $\lceil d(v)/2 \rceil$ and $\lceil d(v)+1/2 \rceil$ respectively, with $d(v)$ standing for the degree of the node $v$ - has been studied. 
Other works carried out the behaviors of majority based systems with respect to the fault tolerance in different networks topologies \cite{Lodi98}, \cite{Santoro03} \cite{Carvaja07}, \cite{Ching09}, \cite{ChoudharyR09}, \cite{Mustafa04} and \cite{MustafaP01}.
The major effort was in determining the distribution of initial faults leading the entire system to a faulty behavior.  Such a pattern, called dynamic monopoly or shortly {\em dynamo} was introduced by Peleg \cite{Peleg96b}, and intensively studied as far as concern bounds of the size of monopolies, time to converge into a fixed point, topologies of the systems (see \cite{BermondBPP96}, \cite{BermondBPP03}, \cite{Bermond98} and \cite{NayakPS92}).
In these works, the propagation of a faulty behavior starts by a well placed set of faulty elements and can be described as a vertex-coloring game on graphs where vertices are colored black (faulty) or white (non-faulty). At each round nodes change their color according to the colors of their neighbors.

In this paper we extend the problem proposed in \cite{Lodi98} by introducing additional elements to the original problem statement. Here, the set of colors is not limited to white or black,
but vertices of the torus can assume a color from a finite set. This general framework finds application in all those contexts where an entity can assume a certain number of different states. 
The rule defined in our model can be stated as follows: a node recolors itself by directly assuming the color of the adjacent vertices either if two neighbors have the same color and the remaining ones have different colors in between or all the neighbors have the same color. Note that our protocol implements an extension of the reverse simple majority rule as defined in \cite{Lodi98}, but in this work we resolve the ambiguity of the rule without stating any preference to any color.
More precisely, in \cite{Lodi98} if in the neighborhood of a node $v$ there are two black and and two white nodes, $v$ recolors black, whereas in our case the node does not change color.
These options aimed at giving a priority to a specific color or to the current color of the vertex in case of ambiguity have been already defined in \cite{Peleg96b} and are known as Prefer-Black (PB) and Prefer-Current (PC), respectively.
Therefore our rule restricted to two colors does not reduce to the rule in \cite{Lodi98}. As a consequence, our results are not trivially obtained by those for two colors. 

We are interested in studying the initial distribution of colors able to produce a monochromatic configuration in a finite time.  After having shortly introduced the notation and the main definitions, 
we present the protocol and the results. Tights lower and upper bounds on the size of dynamos for toroidal mesh, torus cordalis and torus serpentinus are provided.
For each topology a dynamo of the minimum size is given and we prove the time complexity in terms of number of rounds needed to reach a monochromatic configuration.
If the set of colors is ordered (finite set of integers) a different rule can be defined where a node recoloring itself increases its color by one. Studies about the effects of this rule have been addressed in \cite{brunetti2010a,brunetti2010b}.

\section{Basic Definitions}

\subsection{The Interaction Topologies}
Let us consider any $m\times n$ mesh $T=(V,E)$.


\begin{definition} A toroidal  mesh of $m \times n$ vertices is a mesh where each vertex $v_{i,j}$ (
$0\leq i \leq m-1$ and $0\leq j \leq n-1$) is connected to the
four vertices $v_{(i-1)\mod m,j}$ , $v_{(i+1)\mod m,j}$ ,
$v_{i,(j-1)\mod n}$ and $v_{i,(j+1)\mod n}$.
\end{definition}

A \textbf{torus cordalis} is similar to a toroidal mesh except
that the last vertex $v_{{i},{n-1}}$ of each row is connected to
the first vertex $v_{(i+1)\mod m,0}$ of row $i+1$.

A \textbf{torus serpetinus} is similar to a torus cordalis except
that
 the last vertex
$v_{{m-1},{j}}$ of each column $j$ is connected to the first
vertex $v_{0,(j-1)\mod n}$ of column $j-1$.

\subsection{The Protocol}
Let $\mathcal{C}=\{1,\ldots, k\}$ be a finite set of colors. A
coloring of a torus $T$ is a function $r:\; V\rightarrow \mathcal{C}$.
We will refer to a \textbf{bi-colored torus} when $|C| = 2$ and to a \textbf{multi-colored torus} when $|C|>2$.
Let $N(x)$ denote the neighborhood of any vertex $x$ in $V$; we
have that $|N(x)|=4$. Given a coloring $r$ of $V$, we can define
the following ``simple majority with persuadable entities'' protocol (\textbf{SMP-Protocol}).

\begin{center}
\begin{algorithm}[h!]
  \caption{SMP-Protocol}
\begin{algorithmic}[h!]
\FORALL { $x$ $\in$ $T$}
\STATE let $a,b,c,d$ $\in$ $N(x)$
\IF {$(r(a)=r(b))$ $\wedge$ $(r(c)\neq r(d))$ $\vee$ $(r(a)=r(b)=r(c)=r(d))$}
        \STATE  $r(x) \gets r(a)$
\ENDIF
\ENDFOR
\end{algorithmic}
\end{algorithm}
\end{center}

Furthermore we denote the subset of $T$ of all $k$-colored
vertices by $S^k$, and the set of its $k$-colored vertices by
$V^k$ ($k\in \mathcal{C}$).

The recoloring process represents the dynamics of the system.
Depending on the initial coloring of $T$, we get different
dynamics. Among the possible initial configurations (i.e.
assignments of colors) we are interested in those leading to a
monochromatic coloring, so called dynamos.

We generalize the concept of dynamo as follows:

\begin{definition}
Given an initial coloring of $T$, let $S^k$ be a subset of $T$ where all vertices have the same color $k \in \mathcal{C}$;  $S^k$ is a \textbf{dynamo} if a $k$-monochromatic
configuration is reached in a finite number of steps
under the \textbf{SMP-Protocol}.
\end{definition}

We point out that this definition of dynamo is related to the coloring of $T$, i.e. depends both on the initial distribution of the $k$-colored vertices and on the initial assignment of other colors.

\begin{definition}
A dynamo $S^k$ is  \textbf{monotone} if the set of $k$-colored
vertices at time $t$ is a subset of the one at time $t+1$.
\end{definition}
Our focus is in monotone dynamos of minimum size in tori.
Finally we need to introduce the following definitions.
\begin{definition}
A $\textbf{k-block}$ $B^k$ is a connected subset of $T$ made up of
 vertices of the same color $k$ each of which has at least
 \bf{two}
 neighbors in $B^k$.
\end{definition}

Note that vertices in $B^k$ will never change their color.
Moreover, since different tori define different neighborhoods of
the border nodes, we get different configurations for a block in
the three kind of tori. For instance a single column of
$k$-colored vertices is a $k$-block in a toroidal mesh and in a
torus cordalis but not in a torus serpentinus, whereas two
consecutive columns of $k$-colored vertices constitute a $k$-block
in all the tori. A single row of $k$-colored vertices is a
$k$-block in a toroidal mesh but not in a torus cordalis and in a
torus serpentinus, whereas two consecutive rows of $k$-colored
vertices constitute a $k$-block in all the tori. Let
$|\mathcal{C}|>2$, we get:

\begin{definition}
A $\textbf{non-k-block}$ $NB^k$ is a connected subset of $T$ made
up of vertices of colors in $\mathcal{C}\setminus \{k\}$ each of
which has at least \bf{three} neighbors in $NB^k$.
\end{definition}

This definition implies that every vertex in $NB^k$ has at most
one $k$-colored neighbor, namely,
 vertices in $NB^k$ will never recolor by $k$ color. For
 example, two consecutive rows or columns of vertices not colored by $k$ constitute a non-$k$-block in all the tori.

\subsection{The Problem}
In this paper we consider multi-colored tori, and we study the
dynamics of the system under the \textbf{SMP-Protocol}. In
particular, we are interested in determining the minimum size
dynamo for a multi-colored torus (toroidal mesh, torus cordalis
and torus serpentinus). This is obtained by first computing a
lower bound to the size and then an upper bound close to the lower
bound.

Consider the problem of determining a lower bound to the size of a
dynamo for a multi-colored torus. We can define a polynomial time
transformation $\phi: \mathcal{C}\rightarrow \mathcal{C}$ such
that $\phi(i)=1$, for $i=1,\ldots,k-1$, and $\phi(k)=2$. This
transformation allows us to map a multi-colored torus into a
bi-colored torus (where $1$ and $2$ correspond to colors white and
black, respectively). Moreover under transformation $\phi$, a
$non$-$k$-block corresponds to a {\it simple white block} of
\cite{Lodi98}.
\begin{proposition}
A lower bound to the size of a dynamo in a bi-colored torus under
the reverse simple majority rule is a lower bound to the size of a
dynamo in a multi-colored torus under the \textbf{SMP-Protocol}.
\label{prop1}
\end{proposition}
Indeed a lower bound consists in the smallest size of $S^2$
(initial set of black vertices) such that no simple white blocks
can arise in the first problem, and in the smallest size of $S^k$
such that no $non$-$k$-blocks can arise in the second problem.
Because of the correspondence between a $non$-$k$-block and a
simple white block the claim follows. Hence lower bounds for the
multi-colored tori under the \textbf{SMP-Protocol} can be easily
derived by the ones for the bi-colored tori under the reverse
simple majority rule.
\begin{proposition}
An upper bound to the size of a dynamo in a bi-colored torus under
the reverse strong majority rule is an upper bound to the size of
a dynamo in a multi-colored torus under the \textbf{SMP-Protocol}.
\label{prop2}
\end{proposition}
Indeed in order to establish an upper bound to the size of $S^2$,
no {\it strong white blocks} of \cite{Lodi98} have to arise and
successive derived black sets of vertices have to contain the set
$V$ of all the vertices  at the end of the process, in the first
problem. Similarly, to obtain an upper bound to the size of $S^k$,
no $i$-blocks have to arise and successive derived $k$-colored
sets of vertices have to contain the set $V$ of all the vertices
at the end of the process, in the second problem. We have that: a)
strong white blocks correspond to $i$-blocks; b) reverse strong
majority rule is more restrictive than \textbf{SMP-Protocol}:
indeed, under reverse strong majority rule, a vertex recolors
itself if there are three vertices in its neighborhood having the
same color, whereas under the \textbf{SMP-Protocol} two neighbors
with the same color are requested (and the others with different
colors). Because of a) and b) the claim follows. Note that item b)
implies that the upper bound so determined is stronger than
sufficient. This reflects in a bound on the size of $S^k$ far from
the lower bound previously obtained. Therefore the valid upper
bounds close to the lower bounds should be obtained in a different
way and not trivially deduced as a consequences of Proposition
\ref{prop2}.

%
%
%
%
%
%
%
%
%
%
%
%
%
%
%

\section{Dynamos within Persuadable Entities}
We are interested in necessary and sufficient conditions for an
initial set of vertices of the same color to be a minimum size
dynamos. The kind of meshes that we take into consideration are:
the toroidal mesh, the torus cordalis and the torus serpentinus.
As shown above lower bounds for the multi-colored tori under the
\textbf{SMP-Protocol} can be easily derived by the ones for the
bi-colored tori under the reverse simple majority rule. Therefore
we list here the results (Lemmas 1-3, Theorems 1,11,13 of
\cite{Lodi98}) that we need for our further considerations.
%

Let $F\subseteq T$, $R_F$ be the smallest rectangle containing $F$
and $m_F\times n_F$ the size of $R_F$.
Now, the set of vertices derivable from $F$ are the recolored vertices obtained (within a
finite number of steps) by applying the \textbf{SMP-Protocol} to the vertices in $F$.

First of all, we start by looking for the necessary conditions.

\begin{lemma}
Let $S^k$ be a $k$-colored subset of any torus with $m_{S^k}<m-1$
and/or $n_{S^k}<n-1$. Then any $k$-colored set $D^k$ derivable
from $S^k$ satisfies $m_{D^k}\leq m_{S^k}$ and/or $n_{D^k}\leq
n_{S^k}$.
\end{lemma}


This lemma implies that if $S^k$ is a dynamo, then $m_{S^k}\geq
m-1$ and $n_{S^k}\geq n-1$.

\begin{remark}
In a bi-colored toroidal mesh  ($|\mathcal{C}|=2$), the above
lemma holds for $m_{S^2}<m$ and/or $n_{S^2}<n$. Indeed if
$m_{S^2}<m$, then there is at least a row of $1$-colored vertices
$1\neq k$ constituting a $1$-block. As a consequence if $S^2$ is a
dynamo, then $m_{S^2}=m$ and $n_{S^2}=n$. Similarly, in a torus
cordalis we have that $n_{S^2}=n$. \label{rem1}
\end{remark}
This remark stresses that the problem for multi-colored tori under
the \textbf{SMP-Protocol} restricted to two colors is different
from the corresponding problem for bi-colored tori under reverse
simple majority rule.

Furthermore, the following lemma 
holds for all the tori . 

\begin{lemma}
Let $S^k$ be a monotone dynamo and $B^{k}_{i}$ a $k-block$ with $1
\leq i \leq t$. Then,
\begin{itemize}
\item $S^k=B^k_1\cup B^k_2\cup\dots\cup B^k_t$, with $t\geq 1$;
\item $T-S^k$ does not contain any $non$-$k$-block.
\end{itemize}
\label{l2}
\end{lemma}


\subsection{Toroidal Mesh}

In this section we focus on any toroidal mesh $T$ of size $m
 \times n$ with $m,n \geq 2$.

\begin{lemma}
Let $B^k$ be a $k$-block in a colored toroidal mesh of size
$m\times n$. Then,
\begin{itemize} \item if $m_{B^k}=m$ and/or $n_{B^k}=n$, then $|B^k|\geq
m_{B^k}+n_{B^k}-1$; \item if $m_{B^k}\leq m-1$ and/or $n_{B^k}\leq
n-1$, then $|B^k|\geq m_{B^k}+n_{B^k}$.
\end{itemize}
\label{l3}
\end{lemma}


We can state that: 
\begin{theorem}
Let $S^k$ be a monotone dynamo for a colored toroidal mesh of size
$m \times n$. We have
\begin{itemize}
\item (i) $m_{S^k}\geq m-1,\; n_{S^k}\geq n-1$ \item (ii)
$|S^k|\geq m+n-2$.
\end{itemize}
\label{t9}
\end{theorem}

Let us now introduce an example derived by the above mentioned
theorem. Figure \ref{fig:figA} shows an example of a monotone
dynamo $S^k$ of size $m+n-2=16$.

\begin{figure}[h]
 \centering
 \includegraphics[width=60mm]{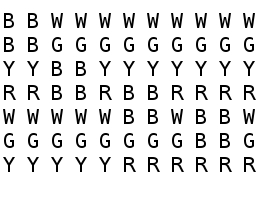}
\caption{A monotone dynamo of black (B) nodes of size $m+n-2$.}
 \label{fig:figA}
\end{figure}



\subsection{On the minimum size dynamos.}

We are going to derive a sufficient condition for an initial set
of $k$-colored vertices to be a monotone dynamo of size close to
the lower bound. So the question is:  how can we construct a
monotone dynamo of the minimum size?

Note that according to Proposition \ref{prop2},
 Theorem 16 of
\cite{Lodi98} gives a dynamo $S^k$ for any multi-colored torus.
Indeed $S^k$ is a $k$-block and $T-S^k$ is a multi-colored forest
such that its leaves have at most one neighbor of their same color
and three $k$-colored neighbors. Therefore no $non$-$k$-block will
be never obtained. Moreover
at each step, leaves recolor themselves constituting the set
derivable by $S^k$, until all the vertices are $k$-colored. This
trivial upper bound is far from the lower bound given in Theorem
\ref{t9}.

\begin{proposition}
 If a minimum size dynamo exists then $|\mathcal{C}|\geq N$, for
$1<N\leq 3$ where $N = min(m,n)$ \label{prep:3}
\end{proposition}

\begin{proof}
Let $N=\min(m,n)$. We notice that there is a link between the
number of colors of $\mathcal{C}$ and the size of $T$ related with the existence of a minimum size
dynamo. Indeed we are looking for a tight bound of the size of a
dynamo. Theorem \ref{t9} states that a necessary condition for
$S^k$ to be a monotone dynamo is that $|S^k|\geq m+n-2$. Let us consider $N>1$, since otherwise we have a dynamo only if
$|\mathcal{C}|=1$. 
Let $N=2$ and without loss of generality we can suppose $N=n$; for
two colors by Remark \ref{rem1} we have that $m_{S^2}=m$ and
$n_{S^2}=2$. Then $|S^2|=m+1=\lceil 2m+1/2\rceil$, and this bound
coincides with the lower bound of Theorem 16 of \cite{Lodi98}. For
more than two colors a column of $k$-colored vertices is a dynamo
of size $m$. If $N=3$, two colors are not enough, since vertices outside a $k$-colored row and column form a
$non$-$k$-block.
\end{proof}


In general we can state the following new theorem that generalizes
and extends some ideas of Theorems 10 and 15 of \cite{Lodi98}. Let
us stress that the bound is tight.
\begin{theorem}
In a multi-colored toroidal mesh of size $m\times n$ with
$|\mathcal{C}|\geq 4$, let $S^k$ be constituted of a $k$-colored
column (row) and a $k$-colored row (column) with one node less;
if, for all $k'\in \mathcal{C}-\{k\}$, $S^{k'}$ is a forest and,
for every vertex $x$ in $V^{k'}$ the vertices in $N(x)\setminus
(V^{k'}\cup V^k)$ have different colors, then $S^k$ is a minimum
size monotone dynamo. \label{theo:10new}
\end{theorem}

\begin{proof}
Indeed $S^k$ is taken according to Theorem \ref{t9} and Lemma
\ref{l2}. Since $S^{k'}$ is a forest, its leaves have at most one
neighbor that is $k'$-colored. In addition,  all the other
neighbors that are not in $V^k$ have different colors. As a
consequence, any vertex will recolor itself just when at least two
of its neighbors are $k$-colored. This prevents to construct any
block during the applications of the \textbf{SMP-Protocol}.
Without loss of generality, let $S^k$ lay on row 0 and column 0
(see Figure \ref{fig:figB} as an example). 
At the first step vertices in positions $(1,1),\; (0,n-1),\; (m-1,1)$ recolor with
$k$-color. Then, $(2,1),\; (1,2),\;(1,n-1),\; (m-2,1),\; (m-1,2),\;(m-1,n-1)$,
and so on until all the vertices are $k$-colored.
\end{proof}


\begin{center}
\begin{figure}[h!]
 \centering
 \includegraphics[width=35mm,bb=0 0 179 174]{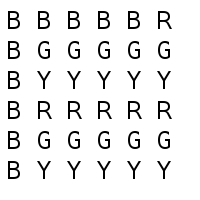}
 \caption{A monotone dynamo of black nodes of size $m+n-2$. Other nodes satisfy the requirement.}
 \label{fig:figB}
\end{figure}
\end{center}






\begin{figure}[h!]
 \centering
 \includegraphics[width=40mm]{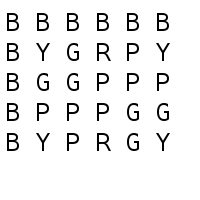}
 \caption{Black  nodes do not constitute a dynamo}
 \label{fig:figD}
\end{figure}



The pattern of four colors in rows $1$ to $m-1$ shown in Figure
\ref{fig:figB} can be repeated several times in order to obtain
a minimum size dynamo in a toroidal mesh of any size. 
By Proposition \ref{prep:3} at least three colors are necessary to have a dynamo, and as we note above
four colors are sufficient to provide a minimum size dynamo. 
Moreover, if $N\geq 4$ less than four colors can not be arranged in order
to satisfy the requirements of Theorem \ref{theo:10new} to have an initial
configuration leading to a monochromatic fixed point.

%
%
Indeed let $k_1$ be the color of a row and a column. If $v_{i,j}$
is colored with $k_2$, in order to minimize the number of colors
consider the following configuration: let $v_{i-1,j}$ and
$v_{i,j-1}$ have color $k_2$: then $v_{i-1,j-1}$ must have a
different color $k_3$ but since it has two neighbors having the
same color $k_2$, this configuration does not verify Theorem
\ref{theo:10new}. So without loss of generality let $v_{i-1,j}$
have color $k_2$ and $v_{i,j-1}$ has color $k_3$. Symmetrically,
with the same arguments we deduce that the color of $v_{i,j+1}$
cannot be $k_2$. In addition color of $v_{i,j+1}$ cannot be $k_3$,
because otherwise $v_{i,j}$ has two neighbors of color $k_3$
contradicting the requirement in Theorem \ref{theo:10new}. Hence
we deduce that $v_{i,j+1}$ is $k_4$-colored.

The same considerations hold in case
of symmetrical configurations.

The requirement of Theorem \ref{theo:10new} cannot be relaxed to
following: {\it for every vertex $x$ in $V^{k'}$ the vertices in
$N(x)\setminus (V^{k'}\cup V^k)$ have different colors, \bf{or} if
any two, say $u,w$ have
the same color, then $y,z\in N(x)\cap V^{k'}$ exist.}\\

In this case it is possible that one of the neighbors of the same
color of $x$ recolors itself by assuming $k$, so that $x$ recolors
itself with $r(u)$. This may lead to a $r(u)$-block. Moreover we
can have configurations that do not allow the propagation of the
color $k$ to all the nodes of $T$. In Figure \ref{fig:figC} we get
an example where no recoloring can arise.

\begin{figure}[h!]
 \centering
 \includegraphics[width=60mm]{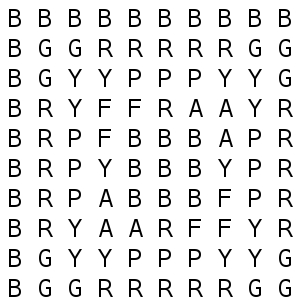}
\caption{Black nodes do not constitute a dynamo.}
 \label{fig:figC}
\end{figure}


Note that such a configuration cannot arise with the constraints
of Theorem \ref{t9}: if $r(v_{2,2})=r(v_{1,2})=r(v_{2,1})$, then
$r(v_{1,1})$ must be different from $r(v_{2,2})$ otherwise we get
a $r(v_{1,1})$-block. But constraints in the theorem forbid two
neighbors of $v_{2,2}$ to have the same color. The same
considerations hold for the vertices $v_{2,n-2}, \;v_{m-2,2}$ and
$v_{m-2,n-2}$.

%


%

\subsection{Torus Cordalis and Torus Serpentinus}

In order to complete our analysis, in this section we are going to
present the remaining cases. Here we provide lower and upper
bounds on the number of $k$-colored nodes leading to new theorems
where the interaction topology of \textbf{SMP-Protocol} is either
a torus cordalis or a torus serpetinus.

\begin{theorem}
 Let $S^k$ be a monotone dynamo for a colored torus cordalis.  Then  $|S^k| \geq n+1$.
\end{theorem}

%

By inspection we can show that the following choice for $S^k$
leads to a monotone dynamo whose cardinality  matches the lower
bound.

\begin{theorem}
In a multi-colored torus cordalis of size $m\times n$ with
$|\mathcal{C}|\geq 4$, let $S^k$ be constituted of $n+1$
$k$-colored vertices placed in a whole row $i$, and one vertex in
row $(i+1)$ mod $m$, column $0$; if, for all $k'\in
\mathcal{C}-\{k\}$, $S^{k'}$  is a forest and, for every vertex
$x$ in $V^{k'}$ the vertices in $N(x)\setminus (V^{k'}\cup V^k)$
have different colors, then $S^k$ is a minimum size monotone
dynamo. \label{theo:12new}
\end{theorem}

Let $N$ be $min (m, n)$. Now we can similarly prove that:

\begin{theorem}
 Let $S^k$ be a monotone dynamo for a colored torus serpentinus. Then $|S^k| \geq
 N+1$.
\end{theorem}

\begin{theorem}
 In a multi-colored torus serpentinus of size $m\times n$ with
$|\mathcal{C}|\geq 4$, let $S^k$ be constituted  of $N+1$
$k$-colored vertices placed in a whole row $i$, and
 in a row $(i+1)$ mod $m$, column $0$, for $N = n$;
 or placed in a whole column $j$, and in column $(j+1)$ mod $n$, row $0$ for $N = m$; if, for all $k'\in \mathcal{C}-\{k\}$,
$S^{k'}$  is a forest and, for every vertex $x$ in $V^{k'}$ the
vertices in $N(x)\setminus (V^{k'}\cup V^k)$ have different colors,
then $S^k$ is a minimum size monotone dynamo. \label{theo:14new}
\end{theorem}

\subsection{On the number of Iterations}
In this section we study the time complexity under the assumptions
that the system is synchronous and takes one unit of time for a
message to arrive and to be processed. Under this restriction we
count the maximum number of rounds needed  for a dynamo to reach
the monochromatic configuration.

Informally, the coloring pattern in case of dynamo in a toroidal
mesh follows an evolution which starts from the corners of the
rectangle  $R_{T-S^k}$ and then it propagates over the diagonals
from the corners to the center. Differently, in case of torus
cordalis and serpentinus, the coloring pattern propagates by rows.

\begin{figure}[h!]
\begin{center}
0 0 0 0 0\\
0 1 2 2 1\\
0 2 3 3 2\\
0 2 3 3 2\\
0 1 2 2 1\\
\caption{A multicolored torus where nodes are represented by the time-steps remaining to assume color $k$.}
 \label{fig:figtm}
\end{center}
\end{figure}

In order to visualize the recoloring pattern evolution we use to
represent positions in a mesh by a matrix $A$ such that
$a_{i,j}=\sigma$ means that vertex $v_{i,j}$ in the mesh recolors
itself by $k$ after $\sigma$ rounds (see Figure \ref{fig:figtm}
for an example). These observations can be formally stated as
follows.

\begin{theorem} Let the initial coloring configuration of a multi-colored toroidal mesh of size $m\times n$ be as in Theorem
\ref{theo:10new}; the number of rounds needed to reach the final
monochromatic configuration is:
\begin{equation}
2max (\lceil((n-1)/2)\rceil -1,\lceil((m-1)/2)\rceil -1)+1 \\
\end{equation}
\end{theorem}
\begin{proof}
In order to count the number of rounds, consider Figure
\ref{fig:figevo}, where at first nodes in positions $(1,1),\;
(1,n-1),\; (m-1,1),\; (m-1,n-1)$ recolor themselves; at second
nodes in
$(2,1),\;(1,2),\;(1,n-2),\;(2,n-1),\;(m-2,1),\;(m-1,2),\;(m-1,n-2),\;(m-2,n-1)$,
and so on. The formula easily follows the inspection.
\end{proof}

\begin{figure}[h!]
\begin{center}
0 0 0 0 0\\
0 1 2 3 4\\
5 6 7 8 7\\
6 7 8 7 6\\
5 4 3 2 1\\
\caption{A multicolored torus where nodes are represented by the time-steps remaining to assume color $k$.}
\label{fig:figevo}
\end{center}
\end{figure}


%

\begin{theorem}
Let the initial coloring configuration of a multi-colored torus
cordalis of size $m\times n$ be as in Theorem \ref{theo:12new};
and let the initial coloring configuration of a multi-colored
torus serpentinus of size $m\times n$ be as in Theorem
\ref{theo:14new} and $N=n$; the number of rounds needed to reach
the final monochromatic configuration is:
\begin{equation}
(\lfloor((m-1)/2)\rfloor -1)n+\lceil(n/2)\rceil \\
\end{equation}
if $m$ is odd;
\begin{equation}
(\lfloor((m-1)/2)\rfloor -1)n+1 \\
\end{equation}
if $m$ is even.
\end{theorem}
\begin{proof}
In order to count the number of rounds, consider Figure
\ref{fig:figevo}: at first nodes in positions $(1,2)$ and
$(m-1,n-1)$ recolor themselves; at second nodes in $(1,3)$ and
$(m-1,n-2)$, and so on until vertices in rows one and $m-1$ have
recolored themselves. Then, vertices of rows $2$ and $m-2$ update
their color, and so on until vertices in positions
$(\lfloor((m-1)/2)\rfloor,1)$ and
$(m-\lfloor((m-1)/2)\rfloor-1,1)$ recolor themselves. Next step
differs in case that $m$ is even or odd. If $m$ is even, then the
two vertices are adjacent since
$(m-\lfloor((m-1)/2)\rfloor-1)-\lfloor((m-1)/2)\rfloor+1=1$. Hence
in one step more the computation ends. Otherwise, if $m$ is odd,
then $(m-\lfloor((m-1)/2)\rfloor-1)-\lfloor((m-1)/2)\rfloor+1=2$:
this means that into $\lceil(n/2)\rceil$ steps the computation
terminates. The formula easily follows the considerations.

The proof follows the steps of the previous one. Indeed even if
the neighbors of the vertices into considerations are different,
the pattern remains exaclty the same.
\end{proof}
%



\section{Conclusions}
%

In this paper we introduced a multicolored version of {\em dynamos}, it is an extension of the original {\em target set selection} problem (TSS) which consists in finding out in a graph the number of nodes needed to activate the entire network in a world of black and white (faulty, non-faulty) nodes.
In this paper, our distributed protocol allows vertices to assume values from a finite set with cardinality higher than 2.
The rule defined in our model can be stated as follows: a node recolors itself by directly
assuming the color of the adjacent vertices either if two neighbors have the same color and the
remaining ones have different colors in between or all the neighbors have the same color.
We derive lower and upper bounds on the size of dynamos for toroidal mesh, torus cordalis and torus serpentinus. The provided bounds are tights. For each one of the addressed topology a dynamo of the minimum size is shown, and the time complexity in terms of number of rounds needed to reach a monochromatic configuration is provided. 

However, at the end of this work, there are some case studies and several pressing and interesting research questions that still remain open. For instance, considering the anologies of this problem with the social processes such as the {\em information diffusion} or {\em viral marketing} processes, scale-free networks could be studied under the \textbf{SMP-Protocol} in order to have a comparitive analysis with respect to other algorithmic models of social influence on social networks, e.g. the bounded confidence model \cite{amblard01}.
In addition, considering the increasing attention to the dynamic aspects of the gossip and of the ``rumors'' spreading, such a protocol should be investigated in contexts where graphs are subject to intermittent availability of both links and nodes \cite{CFQS2010a}. 

%


\bibliographystyle{plain}
\bibliography{biblio}

\end{document}